\documentclass[a4paper,draft]{amsart}
\usepackage[utf8]{inputenc}
\usepackage{amsmath}
\usepackage{amssymb}
\usepackage{amscd}
\usepackage{amsthm}
\usepackage[normalem]{ulem}
\usepackage[centering]{geometry}
\usepackage{amsxtra}
\usepackage{pifont}
\usepackage{mathabx}
\usepackage{xcolor}

\definecolor{darkred}{rgb}{0.5,0,0}
\definecolor{darkgreen}{rgb}{0,0.5,0}
\definecolor{darkblue}{rgb}{0,0,0.5}
\usepackage{mathtools}
\usepackage{shadethm}

\usepackage{microtype}
\usepackage[linktocpage,pdftex]{hyperref}
\hypersetup{linktocpage
           ,pdfborder={0 0 0}
           ,colorlinks
           ,linkcolor=darkblue
           ,filecolor=darkgreen
           ,urlcolor=darkred
           ,citecolor=darkblue
           }

\numberwithin{equation}{section}

\newshadetheorem{boxdef}{Definition}[section]

\DeclareMathOperator{\supp}{supp}
\DeclareMathOperator{\dens}{dens}

\newcommand{\RR}{{\mathbb R}}
\newcommand{\QQ}{{\mathbb Q}}
\newcommand{\ZZ}{{\mathbb Z}}

\newcommand{\TT}{\mathbb T}
\newcommand{\NN}{\mathbb N}
\newcommand{\KK}{\mathbb K}
\newcommand{\XX}{\mathbb X}
\newcommand{\DD}{\mathbb D}

\newcommand{\cL}{{\mathcal L}}

\renewcommand{\hat}{\widehat}

\newcommand{\Qspan}{\mathrm{Span}_{\QQ}}
\newcommand{\Zspan}{\mathrm{Span}_{\ZZ}}

\RequirePackage{ifthen}

\providecommand{\norm}[2][]{\lVert#2\rVert\ifthenelse{\equal{}{#1}}{}{_{#1}}}
\providecommand{\bignorm}[2][]{\bigl\lVert#2\bigr\rVert\ifthenelse{\equal{}{#1}}{}{_{#1}}}
\providecommand{\Bignorm}[2][]{\Bigl\lVert#2\Bigr\rVert\ifthenelse{\equal{}{#1}}{}{_{#1}}}
\providecommand{\biggnorm}[2][]{\biggl\lVert#2\biggr\rVert\ifthenelse{\equal{}{#1}}{}{_{#1}}}
\providecommand{\Biggnorm}[2][]{\Biggl\lVert#2\Biggr\rVert\ifthenelse{\equal{}{#1}}{}{_{#1}}}
\providecommand{\spr}[3][]{\langle#2,#3\rangle\ifthenelse{\equal{}{#1}}{}{_{#1}}}
\providecommand{\bigspr}[3][]{\bigl\langle#2,#3\bigr\rangle\ifthenelse{\equal{}{#1}}{}{_{#1}}}
\providecommand{\Bigspr}[3][]{\Bigl\langle#2,#3\Bigr\rangle\ifthenelse{\equal{}{#1}}{}{_{#1}}}
\providecommand{\biggspr}[3][]{\biggl\langle#2,#3\biggr\rangle\ifthenelse{\equal{}{#1}}{}{_{#1}}}
\providecommand{\Biggspr}[3][]{\Biggl\langle#2,#3\Biggr\rangle\ifthenelse{\equal{}{#1}}{}{_{#1}}}
\let\originald\d 
\renewcommand{\d}{\ifthenelse{\boolean{mmode}}{\mathrm d}{\originald}}
\providecommand{\dd}{\,\d}

\let\originali\i 
\renewcommand{\i}{\ifthenelse{\boolean{mmode}}{\mathrm i}{\originali}}

\newcommand{\Cu}{C_{\mathsf{u}}}
\newcommand{\Cc}{C_{\mathsf{c}}}

\newtheorem{theorem}{Theorem}[section]
\newtheorem{lemma}[theorem]{Lemma}
\newtheorem{prop}[theorem]{Proposition}
\newtheorem{cor}[theorem]{Corollary}

\newtheorem{defi}[theorem]{Definition}

\usepackage{shadethm}

\newtheoremstyle{rremark}%
       {1.8ex\@plus1ex}     
       {2.1ex\@plus1ex\@minus.5ex} 
       {\normalfont}        
       {0pt}                
       {\bfseries}          
       {.}                  
       {.5em}               
       {}                   

\theoremstyle{rremark}
\newtheorem{remark}[theorem]{Remark}

\begin{document}
\title{Inter model sets in $\RR^d$ are model sets}

\dedicatory{We dedicate this work to the memory of Uwe Grimm.}

\author{Christoph Richard}
\address{Department f\"{u}r Mathematik, Friedrich-Alexander-Universit\"{a}t Erlangen-N\"{u}rnberg,
Cauerstrasse 11, 91058 Erlangen, Germany}
\email{christoph.richard@fau.de}

\author{Nicolae Strungaru}
\address{Department of Mathematical Sciences, MacEwan University \\
10700--104 Avenue, Edmonton, AB, T5J 4S2\\
and \\
Institute of Mathematics ``Simon Stoilow''\\
Bucharest, Romania}
\email{strungarun@macewan.ca}

\begin{abstract}  We show that any translate of a model set is a model set in some modified cut-and-project scheme. Restricting to Euclidean direct space, we show that any translate of an inter model set is a model set in some modified cut-and-project scheme with second countable internal space. In both cases, the window in the modified cut-and-project scheme inherits the topological and measure-theoretic properties of the original window. Our results hold in fact for a class beyond inter model sets, which we call almost model sets.
\end{abstract}

\maketitle

\section{Motivation}

A \emph{cut-and-project scheme} $(G, H, \cL)$ consists of locally compact abelian groups $G,H$ and a discrete and cocompact subgroup $\cL \subseteq G \times H$, also called \emph{lattice}, such that $\cL$ projects injectively to $G$ and densely to $H$. Sometimes $G$ is called \emph{direct space}, and $H$ is called  \emph{internal space}.
Denote the coordinate projections by $\pi^G$ and $\pi^H$, respectively. Given any \emph{window} $W \subseteq H$, the \emph{projection set} $\Lambda_W\subseteq G$ is defined by $\Lambda_W=\pi^G(\cL\cap(G\times W))$.
Often the abbreviation $L=\pi^G(\cL)$ is used. Since $\pi^{G}|_{\cL}: \cL \to G$ is one-to-one, the so-called \emph{star-map} $\star: L \to H$, given by
$\star=\pi^{H}\circ (\pi^{G}|_{\cL})^{-1}$, is well-defined. Using the $\star$-map, we may conveniently write $\cL= \{ (g,g^\star) : g \in L \}$ and $\Lambda_W= \{ g \in L : g^\star \in W \}$.
If $W$ is relatively compact and (Borel) measurable, we call $\Lambda_W$ a \emph{weak model set}. If additionally $W$ admits interior points, then $\Lambda_W$ is called a
\emph{model set}. If additionally $W$ has a  topological boundary $\partial W=\overline{W}\setminus W^\circ$ of zero Haar measure, then  $W$ is called \textit{measure-theoretically regular}, and $\Lambda_W$ is a called a  \emph{regular model set}. Background about model sets can be found e.g. in \cite{TAO, RVM3}.

\smallskip

A class of examples in $G=\RR$ arises from  unimodular Pisot substitutions in one dimension with fixed point $\Lambda\subset \RR$, see \cite[Sec.~3]{BG} for a recent exposition. In that setting, the substitution is used in order to construct a cut-and-project scheme $(\RR, H, \cL)$ and an iterated function system in $H$ such that $\Lambda\subseteq \Lambda_W$, where the compact set $W\subseteq H$ is the attractor of the iterated function system. Under additional assumptions on the substitution we may have $\Lambda_{W^\circ} \subseteq \Lambda \subseteq \Lambda_{W}$, where $W$ is measure-theoretically regular. In that case we call $\Lambda$ an \emph{inter regular model set}, compare \cite{JYL07,LM06}. Then $\Lambda$ and the regular model set $\Lambda_W$ coincide up to mismatches of zero density, and they have the same spectral properties. In abuse of notation, some authors call any translate of such $\Lambda$ a regular model set \cite{LM06}. We will study the question whether an inter regular model set may be a regular model set, in some possibly modified cut-and-project scheme.

For illustration consider the Fibonacci substitution $a\mapsto ab$, $b\mapsto a$, compare \cite[Ex.~4.6]{TAO}. Starting with the seed $a|a$, the squared Fibonacci distribution has a unique fixed point. Now replace $a$ and $b$ by intervals of lengths $\tau$ and $1$, respectively, where $\tau=(1+\sqrt 5)/2$ is the golden mean, center the seed at the origin, and consider the collection $\Lambda\subset \ZZ[\tau]$ of all left interval endpoints. A cut-and-project scheme \cite[Ex.~7.3]{TAO} is obtained from the Minkowski embedding of $\ZZ[\tau]$, for which $\Lambda$ is a regular model set with half closed interval as window, while the attractor of the iterated function system gives the closure of the window.
As the associated cut-and-project scheme has an injective $\star$-map, any inter regular model set in the above cut-and-project scheme is in fact a regular model set in the same cut-and-project scheme. A corresponding measure-theoretically regular window is easily constructed, see the proof of Proposition~\ref{lem:imschar}.

\smallskip

Consider now a weak model set $\Lambda_W$ in a general cut-and-project scheme $(G,H,\cL)$. One may pose the question whether there exists a
modified cut-and-project scheme $(G,H',\cL')$ with injective $\star$-map  such that $\Lambda_W=\pi^G(\cL'\cap(G\times W'))$ for some window $W'\subseteq H'$, where $W'$ inherits the following properties $(T)$ and $(M)$ from the window $W$:
\begin{itemize}
  \item{} \emph{topological properties (T)}
  \begin{itemize}
\item precompactness: $\overline W$ is compact in $H$
\item existence of interior points: $W^\circ\neq\varnothing$
\item topological regularity: $\overline W=\overline{W^\circ}$
\end{itemize}
  \item{} {\emph{measure-theoretic properties (M)}}
\begin{itemize}
  \item measure-theoretic regularity: $\partial W=\overline{W}\setminus W^\circ$ has zero Haar measure
 \end{itemize}
\end{itemize}
A partial answer has already been given in \cite[Lem.~4]{LLRSS}, where the construction of a modified cut-and-project scheme uses the Bohr compactification of $G$. While the construction preserves property $(T)$, property $(M)$ may be lost. Also, the construction does not preserve metrisability of internal space. Restricting to Euclidean direct space and second countable internal space, we give an affirmative answer to the above question in Proposition~\ref{prop:interprojection}.

\smallskip

Inter model sets also arise naturally from a dynamical analysis of model sets. To explain this, identify a point set $\Lambda\subseteq G$ by its Dirac comb $\sum_{g\in \Lambda} \delta_g$ and consider its \emph{hull} $\XX(\Lambda)$, i.e., the translation orbit closure of $\Lambda$ with respect to the vague topology. Assume for the following that $\Lambda=\Lambda_W$ is a model set. Then any $\Gamma\in \XX(\Lambda) $ satisfies  $\Lambda_{t+W^\circ} \subseteq s+\Gamma \subseteq \Lambda_{t+\overline{W}}$ for some $(s,t)\in G\times H$, compare Lemma~\ref{lem:tp} below. Whereas such $\Gamma$ is called an inter model set \cite[Eq.~(7)]{BLM}, we will prefer the name \emph{shifted inter model set}, and we will speak of an \emph{inter model set} if $s=0$. One may ask whether any shifted inter model set is an inter model set, possibly in some different cut-and-project scheme. We will give an affirmative answer in Theorem~\ref{thm:shift} below.

Note that $(s,t)$ is uniquely determined modulo $\cL$ if the window is non-empty and aperiodic, i.e., $t+W=W$ implies $t=0$, compare  \cite{BLM}. In that case any $\Gamma\in \XX(\Lambda)$ can be parametrised by $(s,t)+\cL$. This gives rise to the so-called torus parametrisation, see \cite{KR18,KR19} for a recent approach. Again, the question arises whether an element of $\XX(\Lambda)$ may be a (shifted) model set.
%
If the translated window $t+W$ is in generic position, i.e., if the boundary of $t+W$ does not intersect the projected lattice, then we have $\Lambda_{t+W^\circ} =\Lambda_{t+\overline{W}}=\Lambda_{t+W}$. Therefore, if $\Gamma$ is parametrised by $(s,t)+\cL$ for $t+W$ in generic position, $s+\Gamma=\Lambda_{t+W}$ implies that $\Gamma$ is a shifted model set and hence a model set. We show in Proposition~\ref{prop:hull-reg-is-reg} that the latter conclusion holds for all $\Gamma\in \XX(\Lambda)$.



\smallskip

\section{Setting and main results}\label{sec:main}


In order to formulate our results,  we use the following setting.  We allow for a shift of the lattice  in the cut-and-project construction, which leads to  \emph{shifted projection sets} defined by
\begin{displaymath}
\Lambda_W(x)=\pi^G((x+\cL)\cap(G\times W)) \ .
\end{displaymath}
Note that shifting the lattice is equivalent to translating both the projection set and the window. Indeed, for $x=(s,t)\in G\times H$ we have $\Lambda_W(x)=s+\Lambda_{-t+W}$. Recall that, in abuse of notation, $\Lambda_W(x)$ is sometimes called a model set if $\Lambda_W$ is a model set.
In a given cut-and-project scheme $(G,H,\cL)$ with window $W\subseteq H$, one may consider a \emph{shifted inter projection set} $\Gamma\subseteq G$ satisfying
\begin{displaymath}
\Lambda_{W^\circ}(x)\subseteq \Gamma \subseteq \Lambda_{\overline{W}}(x)
\end{displaymath}
for some $x\in G\times H$.

\begin{remark}\label{rem:sips}
As discussed in the previous section, shifted inter projection sets arise naturally when studying projection sets from a dynamical perspective.
If $\Lambda_W$ is a projection set, then every $\Gamma\in \XX(\Lambda_W)$ is a shifted inter projection set. A proof, which extends previous results from \cite{sch00, BLM, KR18}, will be given in Lemma~\ref{lem:tp} below.
\end{remark}


\smallskip

We argue first that shifting is irrelevant, i.e.,  any translate of an inter model set is an inter model set in some modified cut-and-project scheme, such that properties $(T)$ and $(M)$ of the window are preserved. This is the content of the following statement.

\begin{theorem}\label{thm:shift}
Consider a cut-and-project scheme $(G,H,\cL)$ and some $a \in G$. Then there exist a cut-and-project scheme $(G, H', \cL')$ and some $b \in H'$ such that
\begin{itemize}
\item[(a)] $H$ is an open and closed subgroup of $H'$.
  \item[(b)] $(a,b) \in \cL'$.
  \item[(c)] $\cL' \cap (G \times H)= \cL$.
  \item[(d)] For each $W \subseteq H$ we have $\Lambda_W=\pi^G(\cL'\cap(G\times W))$ and $a+\Lambda_W=\pi^G(\cL'\cap(G\times (b+W)))$.
  \item[(e)] Both $W\subseteq H'$ and $b+W\subseteq H'$  inherit any property in $(T)$ or $(M)$ from $W\subseteq H$.
  \item[(f)] $H'$ inherits metrisability from $H$.
\end{itemize}
\end{theorem}

\begin{remark}
As a consequence of the above theorem, any shifted weak model set, shifted model set, shifted regular model set and shifted inter model set, respectively, is a weak model set, model set, regular model set and inter model set, respectively. Moreover, the same will be true for almost regular model sets which we introduce after Eqn.~\eqref{eq1}: shifted almost regular model sets are almost regular model sets.
\end{remark}

Note that the projection set in the above proposition may arise from a different cut-and-project scheme as the shifted projection set. We will study so-called translation cut-and-project schemes in  Section \ref{tcps}, where the above result is a direct consequence of Proposition~\ref{prop2}. Our construction is inspired by \cite[Thm.~6.1]{KS}.

\smallskip

One may further ask whether an inter projection set may be a projection set, such that properties $(T)$ and $(M)$ of the window are preserved. As pointed out above, the answer is affirmative for cut-and-project schemes where $\cL$ projects injectively to $H$, meaning that the $\star$-map of the cut-and-project scheme is one-to-one. Indeed, we have:

\begin{prop}\label{lem:imschar}
Let $(G,H,\cL)$ be a cut-and-project scheme with injective $\star$-map. Then any inter projection set in $(G,H,\cL)$ is a projection set in $(G,H,\cL)$. Moreover the window of the projection set inherits properties $(T)$ and $(M)$. Assume additionally that $G$ is $\sigma$-compact and that the window $W$ of the inter projection set is precompact and measurable. Then the window of the projection set is precompact and measurable.
\end{prop}

\begin{proof}
Consider a cut-and-project scheme $(G,H,\cL)$ such that $\cL$ projects injectively to $H$. Consider $\Gamma\subseteq G$ and $W\subseteq H$ such that $\Lambda_{W^\circ}\subseteq \Gamma \subseteq \Lambda_{\overline{W}}$.
Due to injectivity of the $\star$-map, we have $\Gamma=\Lambda_{\Gamma^\star}$ for any $\Gamma\subseteq L$. Then $W'=W^\circ \cup \Gamma^\star$ satisfies
$\Lambda_{W'}=\Lambda_{W^\circ}\cup \Lambda_{\Gamma^\star}=\Lambda_{W^\circ}\cup \Gamma=\Gamma$.
Let us analyse the window properties $(T)$ and $(M)$. Note that the definition of $W'$  and $\Lambda_{W^\circ}\subseteq \Gamma\subseteq \Lambda_{\overline{W}}$  give $W^\circ  \subseteq W' \subseteq \overline{W}$, which implies
\begin{equation}\label{eq11}
W^\circ \subseteq (W')^\circ \subseteq W' \subseteq \overline{W'} \subseteq \overline{W} \ .
\end{equation}
\noindent Assume that $W$ is precompact. Then $\overline{W}$ is compact. Therefore $W'$ is precompact by \eqref{eq11}.

\noindent Assume $W^\circ \neq \varnothing$. Then  $(W')^\circ \neq \varnothing$ follows from \eqref{eq11}.

\noindent Assume $\overline{W}=\overline{W^\circ}$. Together with  \eqref{eq11} this gives
$\overline{W}  = \overline{W^\circ } \subseteq \overline{W'} \subseteq \overline{W}$.
We thus can conclude $\overline{W} = \overline{W'} = \overline{W^\circ }  \subseteq  \overline{(W')^\circ}$, and the claim follows.

\noindent Assume that $W$ is measure-theoretically regular.  By \eqref{eq11} we have
$\partial(W') \subseteq \overline{W}  \backslash W^\circ =\partial W$.
Thus $W'$ is measure-theoretically regular.

\noindent Finally, assume that $G$ is $\sigma$-compact and that $W$ is precompact and measurable. Then $W'$ is precompact by the previous argument, and $W'$ measurable as $\Gamma$ is countable, due to uniform discreteness of $\Lambda_{\overline{W}}$ and $\sigma$-compactness of $G$.
\end{proof}

\begin{remark}
Examples of setting in Proposition~\ref{lem:imschar} are complete cut-and-project schemes, i.e., cut-and-project schemes where $\cL$ projects both injectively to $H$ and densely to $H$. For $G=\mathbb \RR^d$ and $H=\mathbb \RR^m \times \TT^n \times \DD$, where $\DD$ is a countable discrete group, complete cut-and-project schemes have recently been classified \cite[Thm.~3]{AACM19}.
\end{remark}

In Section \ref{sec:ecps} we study when a projection set may arise from a cut-and-project scheme $(G,H,\cL)$ with injective $\star$-map. We will construct such a cut-and-project scheme that also preserves properties $(T)$ and $(M)$ of a window. Our construction assumes Euclidean direct space $G$ and assumes  $H$ to be second countable. The following result is a reformulation of Theorem~\ref{thm2} below.




\begin{prop}\label{prop:interprojection}
Consider a cut-and-project scheme $(\RR^d,H,\cL)$ with second countable internal space $H$. Then any inter projection set in $(\RR^d,H,\cL)$ is a projection set, and the window of the projection set inherits properties $(T)$ and $(M)$ from the window of the inter projection set. Assume additionally that the window $W$ related to the inter projection set is precompact and measurable. Then the window of the projection set is precompact and measurable.
\end{prop}

In order to refine the above result, we investigate when the internal space can be assumed to be second countable. This is indeed the case for regular model sets. In fact, we prove a stronger statement in Proposition~\ref{mainsecond} below, by enlarging the class of inter model sets as follows.

\begin{defi}We say that $\Gamma \subseteq G$ is an \emph{almost model set} if there exist a cut-and-project scheme $(G, H, \cL)$ and sets $U \subseteq W \subseteq H$ such that $U$ is non-empty open, $W$ is compact,  and
\begin{equation}\label{eq1}
\Lambda_U \subseteq \Gamma \subseteq \Lambda_W \ .
\end{equation}
If additionally $W \backslash U$ has vanishing Haar measure, we call $\Gamma$ an \emph{almost regular model set}.
We will also refer to $\Gamma$ as being a \emph{$(U,W)$-almost model set}.

\end{defi}

\begin{remark}
\begin{itemize}
  \item[(a)] If $\Gamma$ is an inter (regular) model set, then $\Gamma$ is an almost (regular) model set.
  \item[(b)] For an $(U,W)$-almost model set $\Gamma$, the density formula for weak model sets \cite[Prop.~3.4]{HuckRichard15} combined with \eqref{eq1} gives
\begin{displaymath}
\begin{split}
\dens(\cL) &\cdot m_H(U) \leq \underline{\dens}(\Lambda_U) \leq  \underline{\dens}(\Gamma) \\
&\leq  \overline{\dens}(\Gamma) \leq  \overline{\dens}(\Lambda_W) \leq \dens(\cL) \cdot  m_H(W) \ .
\end{split}
\end{displaymath}
Here $m_H$ denotes a Haar measure on $H$, $\dens(\cL)$ denotes the inverse measure of a measurable fundamental domain of $\cL$, and $ \underline{\dens}$, $\overline{\dens}$ denote the lower and upper density along a given averaging sequence.
In particular, any almost regular model set $\Gamma$ differs from the regular model set $\Lambda_W$ on a set of density zero. This readily implies that any almost regular model set $\Gamma$ is Weyl almost periodic \cite{LSS20}, has the same autocorrelation, diffraction and Fourier--Bohr coefficients as
$\Lambda_W$, and that its hull, equipped with its natural translation action and its unique ergodic probability measure, is measure-theoretically isomorphic to that of $\Lambda_W$.
\end{itemize}
\end{remark}

We can now state our main result, whose proof will be given at the end  of Section~\ref{sec:ecps}.
%
\begin{theorem}\label{thm:mainintro}
Let $\Lambda \subseteq \RR^d$ be any shifted almost model set from some cut-and-project scheme. Then, $\Lambda$ is a model set in some modified cut-and-project scheme with second countable internal space and injective $\star$-map, such that measurability and properties $(T)$ and $(M)$ of the window are preserved.
In particular, any translate of an inter (regular) model set in Euclidean space is a (regular) model set.
\end{theorem}


Let us point out an application of Theorem~\ref{thm:mainintro} which is related to point sets obtained from substitution. We cite Theorem 5.3 from \cite{JYL07}, where we restrict to Delone sets instead of Delone multisets for convenience. Recall that a Delone set is a point set that is uniformly discrete and relatively dense. For the notion of finite local complexity, see the remark before Proposition~\ref{prop:flcrep}.
\begin{theorem}[\cite{JYL07}] Let $ \Lambda\subset \RR^d$ be a primitive substitution Delone set such that every
$\Lambda$-cluster is legal and such that $\Lambda$ has finite local complexity. Then $\Lambda$ has pure point dynamical spectrum
if and only if $\Lambda$ is a shifted inter model set with topologically regular window. \qed
\end{theorem}
Due to our results, we may replace ``shifted inter model set'' by ``model set''.

\section{Translation cut-and-project schemes}\label{tcps}

Consider a projection set $\Lambda_W$ in some cut-and-project scheme $(G,H, \cL)$ with window $W\subseteq H$. Given some $a \in G$, we ask whether $a+\Lambda_W$ is a projection set. The answer is affirmative for $a \in \pi_G(\cL)$, as we then have $a+\Lambda_W=\Lambda_{a^\star+W}$. In general $a+\Lambda_W$ is a projection set in some modified cut-and-project set $(G, H', \cL')$, where $H'$ is a factor group of $H \times \mathbb Z$. Our construction is inspired by  \cite[Thm.~6.1]{KS}.  Fix nonzero $a\in G$ and consider the group $\langle a \rangle=\Zspan(\{a\})$ generated by $a$. We first consider the case $\langle a \rangle \cap \pi_G(\cL)= \{ 0 \}$, which we call the incommensurate case.

\subsection{The incommensurate case}

A natural construction leads to the following result.

\begin{prop}[translation CPS]\label{prop1} Consider a cut-and-project scheme $(G, H, \cL)$. Fix nonzero $a \in G$ so that $\langle a \rangle \cap  \pi_G(\cL)=\{0\}$.
Define $H' = H \times \ZZ$ and $\cL'= (\cL\times \{0\})+\langle (a,0,1)\rangle$.
Then the following hold.
\begin{itemize}
  \item[(a)] $(G, H', \cL')$ is a cut-and-project scheme.  If $\cL$ is countable, then so is $\cL'$.
  \item[(b)] For $W\subseteq H$ and $n\in \ZZ$ we have $na+\Lambda_W=\Lambda'_{W'}$, where $W'=W\times \{n\}$.
  \item[(c)] $H$ is homeomorphic to $H \times  \{0 \}$ and $H \times \{ n \}$ is open and closed in $H'$.
  \item[(d)] The window $W'=(W\times \{n\})\subseteq H'$ inherits properties $(T)$ and $(M)$ from $W\subseteq H$. If $W$ is measurable, then $W'$ is measurable.
\end{itemize}
\end{prop}

\begin{proof}
``(a)'' By definition, countability of $\cL'$ follows from countability of $\cL$. It is clear that $\cL'$ is a group. In order to show that $\cL'$ is a lattice, it suffices to show that $\cL'$ is a Delone set. Since $\cL$ is a lattice in $G \times H$, there exists a compact set $K \subseteq G\times H$ such that $\cL+K=G\times H$. This implies
\begin{displaymath}
\cL'+(K\times \{0\})=(\cL\times \{0\})+\langle (a,0,1)\rangle+(K\times \{0\})=(G\times H\times\{0\})+\langle (a,0,1)\rangle=G\times H\times \mathbb Z =G\times H' \ .
\end{displaymath}
Since $\cL$ is uniformly discrete, there exists a zero neighborhood $U\subseteq G\times H$ such that $\cL \cap U = \{ (0,0) \}$. Then the open zero neighborhood $U\times \{0\}\subseteq G\times H'$ satisfies
\begin{displaymath}
\cL' \cap (U\times \{0\}) = ((\cL\times \{0\})+\langle (a,0,1)\rangle)  \cap (U\times \{0\})
=  (\cL\times \{0\})\cap (U\times \{0\}) = \{(0,0,0)\} \ .
\end{displaymath}

\smallskip

\noindent We show next that $\pi^{G}: G\times H'\to G$ is one-to-one when restricted to $\cL'$. Assume $\ell'=(\ell,0)+n(a,0,1)\in\cL'$ for some $\ell\in \cL$ and $n\in \ZZ$ such that $\pi^G(\ell')=na+\pi^G(\ell)=0$. By assumption, we get $n=0$ and $\pi^G(\ell)=0$. As $\pi^{G}: G\times H\to G$ is one-to-one when restricted to $\cL$, we get $\ell=0$, which implies $\ell'=(\ell,0)+n(a,0,1)=(0,0,0)$. This proves the claim.

\smallskip

\noindent Denseness of  $\pi^{H'}(\cL')$ in $H'$ is inherited from $(G,H,\cL)$. Indeed, consider any non-empty open set $U \subseteq H'$. Then there exist some non-empty open set $V  \subseteq H$ and $n \in \ZZ$ such that $V \times  \{ n \} \subseteq U$.
Since $(G,H ,\cL)$ is a cut-and-project scheme, there exists some $\ell \in \cL$ such that $\pi^H(\ell)\in V$. We then have $\pi^{H'}((\ell,0)+n(a,0,1))=(\pi^H(\ell),0)+(0,n)=(\pi^H(\ell),n) \in \pi^{H'}(\cL')\cap (V\times \{n\})\subseteq \pi^{H'}(\cL')\cap U$. This proves (a).

\smallskip

\noindent ``(b)'' follows by direct calculation. Let $W'=W\times \{n\}\in H'$. Then
\begin{displaymath}
\begin{split}
\Lambda'_{W'}&=\pi^G(\cL'\cap (G\times W'))= \pi^G(((\cL\times \{0\})+\langle (a,0,1)\rangle) \cap (G\times W\times \{n\}))\\
&= \pi^G(((\cL\times \{0\})+n (a,0,1)) \cap (G\times W\times \{n\})) = \pi^G(\cL\cap (G\times W))+na=\Lambda_W + na \ .
\end{split}
\end{displaymath}
\noindent ``(c)'' is obvious, and ``(d)'' follows from ``(c)''.
\end{proof}

\subsection{The general case}\label{rem:notinj}

The lattice $\cL'$ in the above construction does no longer project injectively to $G$ if $\langle a \rangle\cap \pi_G(\cL)\ne \{0\}$. Indeed, let $m\in \mathbb N$ be minimal such that $b=ma\in \pi_G(\cL)$, and consider arbitrary $(na+x,x^\star, n)\in \cL'$. Then $na+x=0$ implies $x=sb$ and $n=-sm$ for some $s\in \ZZ$. We thus obtain
$(na+x,x^\star,n)= (0,sb^\star,-sm)$, which violates injectivity. But injectivity might be restored from factoring by the subgroup $J=\langle (b^\star, -m)\rangle\subseteq H'$. Consider thus the factor group $H'/J$ equipped with the quotient topology. Take any $(h,n)\in H'$. Writing $n=r+sm$ for $r\in \{0,\ldots, m-1\}$ and $s\in \mathbb Z$ and noting $(-sb^\star, sm)\in J$, we infer $(h,n)+J=(h+sb^\star,r)+J$. This means that the factor group $H'/J$ is group-theoretically isomorphic to $(H\times \ZZ_m, \oplus)$, where
$$
(h_1,r_1) \oplus (h_2,r_2)= \left\{
\begin{array}{cc}
(h_1+h_2, r_1+r_2)  & \mbox{ if } r_1+r_2 <m  \\
(h_1+h_2+b^\star, r_1+r_2-m)  & \mbox{ if } r_1+r_2  \geq m  \\
\end{array}
\right. \ .
$$
Moreover, under the above identification, the induced topology on $H\times \ZZ_m$  is the product topology.

\medskip

The previous discussion suggests that we might obtain a translation cut-and-project scheme for the general case by a factor group construction. Quotient cut-and-project schemes have been intensively studied in \cite[Sec.~6]{KR19} by factoring with compact subgroups. Here the situation is different as $J\subseteq H'$ may not be compact. Whereas most arguments from \cite{KR19} apply in the present case, uniform discreteness of the quotient lattice has to be shown separately. For the ease of the reader, we provide complete arguments.

\begin{prop}[extended quotient CPS]\label{prop2} Let $(G, H, \cL)$ be a cut-and-project scheme. Consider nonzero $a \in G$. If $\langle a \rangle\cap \pi_G(\cL)\ne \{0\}$, let $m \in \NN$ be minimal so that $b=ma \in \pi_G(\cL)$. Otherwise set $m=0$ and $b=0$.  Consider the group $J=\langle  (b^\star, -m)\rangle \subseteq H\times \ZZ$ equipped with the discrete topology. Define $H'=(H\times \ZZ)/J$ and $\cL'=(\cL\times \{0\})+\langle (a,0,1)\rangle + (\{0\}\times J)\subseteq G\times H'$.
Then the following hold.
\begin{itemize}
  \item[(a)] $(G, H', \cL')$ is a cut-and-project scheme.  If $\cL$ is countable, then so is $\cL'$.

   \item[(b)] For $W\subseteq H$ and $n\in \ZZ$ we have $na+\Lambda_W=\Lambda'_{W'}$, where $W'=(W\times \{n\})+J$.
   \item[(c)] The map $x \mapsto (x,0)+J$ is a continuous open embedding of $H$ into $H'$.
    \item[(d)] The window $W'=((W\times \{n\})+J)\subseteq H'$ inherits properties $(T)$ and $(M)$ from $W\subseteq H$. If $W$ is measurable, then $W'$ is measurable.
\end{itemize}
\end{prop}

\begin{remark}
Note that Proposition~\ref{prop1} is a special case of Proposition~\ref{prop2} with trivial group $J$. This justifies the same notation for the two cut-and-project schemes.
\end{remark}

\begin{proof}
We use the canonical identification $(G\times H\times \ZZ)/(\{0\}\times J)\cong G\times H'$, compare \cite[Sec.~6]{KR19}.

\noindent ``(a)''
By definition, countability of $\cL'$ follows from countability of $\cL$. We start by showing that $\cL'$ is a lattice in $G \times H'$. Note first that $\cL'$ is a group as $\cL'=\cL+\langle (a,0,1)\rangle +(\{0\}\times J)$ and as $\cL$ is a group. The group $\cL'$ is also relatively dense: Take $K\in G\times H$ such that $\cL +K=G\times H$, and define $K'=(K\times \{0\})+(\{0\}\times J)$. We then have
\begin{displaymath}
\begin{split}
\cL'+K'&=(\cL+\langle (a,0,1)\rangle+(\{0\}\times J))+((K\times \{0\})+(\{0\}\times J)) \\
&= \cL +\langle (a,0,1)\rangle+(K\times \{0\})+(\{0\}\times J)\\
&=  (G\times H\times \{0\}) +\langle (a,0,1)\rangle+(\{0\}\times J)
=G\times H'\ .
\end{split}
\end{displaymath}

\noindent Next, we prove uniform discreteness of $\cL'$. Since $\cL$ is a lattice in $G \times H$ there exists a zero neighborhood  $U \subseteq G\times H$ such that $\cL \cap U = \{(0,0) \}$. Define $U'\in G\times H'$ by $U'=(U\times\{0\})+(\{0\}\times J)$. Recalling $J=\langle  (b^\star, -m)\rangle$ and $(ma,(ma)^\star)=(b,b^\star)\in\cL$, we then have
\begin{displaymath}
\begin{split}
\cL'\cap U' &= ((\cL\times \{0\})+ (na+x,x^\star,n) + (\{0\}\times J)) \cap ((U\times\{0\})+(\{0\}\times J)) \\
&=  ((\cL\times \{0\})+ \langle m(a,0,1)\rangle + (\{0\}\times J)) \cap ((U\times\{0\})+(\{0\}\times J)) \\
&=  ((\cL\times \{0\})+ (\{0\}\times J) + (\{0\}\times J)) \cap ((U\times\{0\})+(\{0\}\times J))  \\
&=  ((\cL\times \{0\})+ (\{0\}\times J)) \cap ((U\times\{0\})+(\{0\}\times J))\\
&=  (((\cL\cap U)\times\{0\})+(\{0\}\times J)) =  \{0\}\times J \ .
\end{split}
\end{displaymath}
In the second equality, we used that the $\ZZ$-component is restricted to be an integer multiple of $m$. In the third equality, we used $sm(a,0,1)+(\cL\times \{0\})=s(0,-b^\star,m)+(\cL\times \{0\})$ for any $s\in \ZZ$.
The above proves that $\cL'$ is a lattice in $G \times H'$.

\smallskip

\noindent We show that $\cL'$ projects injectively to $G$. Let $\ell'\in\cL'$ such that $\pi^G(\ell')=0$. Write $\ell'= (na+x,x^\star,n) +(\{0\}\times J)$ for $(x,x^\star)\in\cL$ and $n\in\ZZ$. Then $0=\pi^G(\ell')=na+x$. By the remark at the beginning of Section~\ref{rem:notinj}  we thus have $(na+x,x^\star,n)\in (0\times J)$, which implies $\ell'=(\{0\}\times J)$.

\smallskip

\noindent We also have that $\pi^{H'}(\cL')$ is dense in $H'$. Indeed we have
\begin{displaymath}
\begin{split}
\overline{\pi^{H'}(\cL')}&=\overline{\pi^{H'}(\cL+ \langle (a,0,1)\rangle +(\{0\}\times J))}=\overline{\pi^{(H\times \ZZ)}(\cL +  \langle (a,0,1)\rangle )+J}\\
&\supseteq\overline{\pi^{(H\times \ZZ)}(\cL + \langle (a,0,1)\rangle)}+J=(H\times \ZZ)+J=H' \ ,
\end{split}
\end{displaymath}
where we used continuity of the quotient map for the above inclusion.
This proves (a).

\noindent ``(b)'' follows by direct calculation. Let $W'=((W\times \{n\})+J)\in H'$. Then
\begin{displaymath}
\begin{split}
\pi^G(\cL&'\cap (G\times W'))\\
&= \pi^G(((\cL\times \{0\})+\langle (a,0,1)\rangle) +(\{0\}\times J)) \cap (G\times ((W\times \{n\})+J)))\\
&= \pi^G(((\cL\times \{0\})+n (a,0,1)+(\{0\}\times J)) \cap (G\times ((W\times \{n\})+J))) \\
&= na+\pi^G(((\cL\times \{0\})+(\{0\}\times J)) \cap (G\times ((W\times \{0\})+J))) \\
&=na+ \pi^G(\cL\cap (G\times W))=na + \Lambda_W \ .
\end{split}
\end{displaymath}
Here the third equation follows using arguments similar to the above proof uniform discreteness, and the fifth equation uses the map $\ell\mapsto  \ell \times\{0\}+\{0\}\times J$, which yields a bijection between the corresponding subsets of $G\times H$ and $G\times H'$.

\noindent ``(c)'' Consider the canonical group homomorphism $f:H\to H'$ given by $h\mapsto (h,0)+J$. Since $\{ 0\}$ is open in $\ZZ$, it is immediate that the map $H \ni x \mapsto (x,0) \in H \times  \{0\}$ is a continuous embedding of $H$ into $H \times \ZZ$. Next, the canonical projection $\pi : H \times \ZZ \to H \times \ZZ /J$ is a continuous group homomorphism which is an open mapping \cite[p.~59]{RS00}. Therefore, the map $f$ is a continuous open group homomorphism.
Finally, it follows immediately from the definition of $J$ that the restriction of $\pi$ to $H \times \{0 \}$ is one to one, which gives that $f$ is one to one and hence an embedding.

\noindent ``(d)'' By ``(c)'', $f(H)\subseteq H'$ is homeomorphic to $H$ and open in $H'$. Thus if $W$ is precompact, then $f(W)$ is precompact in $f(H)$, which by definition of the trace topology implies that $f(W)$ is precompact in $H'$.
Similar arguments apply to the case to $W$ having non-empty interior and $W$ being topologically regular. Also note that $f(W)$ is measurable for measurable $W$, as the Borel $\sigma$-algebra on $f(H)$ is contained in the Borel $\sigma$-algebra on $H'$.
Finally observe $\partial (f(W)) \subseteq f(\partial W)$, which holds as $f$ is open and closed (for the latter statement note that $f^{-1}|_{f(H)}$ is continuous). Hence $m_{H'}(\partial f(W))\le m_{H'}(f(\partial W))$, where $m_{H'}$ is a Haar measure on $H'$. Noting that $A\mapsto m_{H'}(f(A))$ defines a Haar measure on $H$ as $m_{H'}(f(H))>0$, measure-theoretic regularity of $f(W)$ follows from measure-theoretic regularity of $W$.
\end{proof}

We complete the section by  proving Theorem~\ref{thm:shift}.

\begin{proof}[Proof of Theorem~\ref{thm:shift}]
In the notation from the proof of Proposition~\ref{prop2}, since $f : H \to H'$ is an embedding, we can assume without loss of generality that $H$ is a subgroup of $H'$.

\noindent ``(a)'' and ``(b)'' follow now from the previous result.

\noindent ``(c)'' follows immediately from the definition of $J$, and ``(d)'' is obvious.

\noindent ``(e)'' follows from Proposition~\ref{prop2} (d).

\noindent ``(f)'' is an immediate consequence of the construction of $H'$, but also follows alternatively from the fact that if $H$ is metrisable, then $H$ ia a metrisable neighborhood of zero in $H'$ and hence $H'$ is metrisable.
\end{proof}

\section{Torus parametrisation and the almost model set hull}

We start with the following preparatory lemma, which generalises previous results \cite{sch00, BLM, KR18}.

\begin{lemma}\label{lem:tp}
Fix some cut-and-project scheme $(G, H, \cL)$, and take open $U\subseteq H$ and compact $W\subseteq H$ such that $U\subseteq W$. Consider any $(U,W)$-almost model set $\Gamma\subseteq G$ . Let $(s_\alpha)$ be a net in $G$ such that
$(s_{\alpha}+\Gamma)$ converges vaguely to some point set $\Gamma'$ and that $((-s_\alpha, 0)+\cL)$ converges to some $(s,t)+\cL\in   (G \times H)/\cL$.
Then $\Lambda_{U+t} \subseteq s+\Gamma' \subseteq \Lambda_{W+t}$.
\end{lemma}

\begin{proof}
Note first that if $\Gamma'= \varnothing$ then $\Gamma$ is not relatively dense, and hence neither is
$\Lambda_U$, which gives that $U=\varnothing$. Therefore, the inclusions trivially hold. We can thus focus on the case $\Gamma' \neq \varnothing$. Fix some open zero neighborhood $V$ so that $\Gamma$ is $V$-uniformly discrete, i.e, $(\Gamma+x)\cap V$ contains at most one point, for all $x\in G$. Fix some $a \in \Gamma'$. 
Since $\Gamma$ is a subset of a model set and thus has finite local complexity, vague convergence means that
for any compact $K \subseteq G$ there exist a net $(s_{\alpha,K}')$ converging to $0$, with $s_{\alpha,K}' \in V$ for all $\alpha$, and some index $\alpha_K$ such that
$(s_\alpha+s_{\alpha,K}' +\Gamma)\cap (K \cup \{ a \}) =\Gamma' \cap (K \cup \{ a \})$ for all $\alpha >\alpha_K$, compare \cite{BL}. In particular,
\[
a \in s_\alpha+s_{\alpha,K}' +\Gamma \Rightarrow s_{\alpha,K}' \in \left( s_\alpha+\Gamma-a\right) \cap V \,.
\]
The $V$-uniform discreteness of $\Gamma$ implies that there exists a unique element in $\left( s_\alpha+\Gamma-a\right) \cap V$ and therefore, $s_{\alpha,K}'$ is independent of $K$. This shows that there exists some a net $(s_{\alpha}')$ converging to $0$, such that, for all compact $K$ there exists some index $\alpha_K$ such that
$(s_\alpha+s_{\alpha}' +\Gamma)\cap (K \cup \{ a \}) =\Gamma' \cap (K \cup \{ a \}) $ for all $\alpha >\alpha_K$. In particular we have
\[
(s_\alpha+s_{\alpha}' +\Gamma)\cap K =\Gamma' \cap K 
\]
for all $\alpha > \alpha_K$.
Thus, by replacing $s_{\alpha}$ by $s_\alpha+s_{\alpha}'$, we may assume without loss of generality that $(s_\alpha+\Gamma)\cap K =\Gamma' \cap K$ for all $\alpha >\alpha_K$.

\smallskip

\noindent Let us first prove the second inclusion. Fix $x \in \Gamma'$ and pick some compact $K$ such that
$ x \in \Gamma' \cap K$.
Then $x \in (s_\alpha+\Gamma)\cap K \subseteq s_\alpha+ \Lambda_{W}$ for all $\alpha > \alpha_K$.
For any $\alpha > \alpha_K$ take some $x_{\alpha} \in \Lambda_{W}$ such that $x=s_\alpha+x_\alpha$.
By compactness of $W$, there exists some subnet $(x_{\beta})$ of $(x_{\alpha})$, such that $(x_{\beta}^\star)$ converges in $H$ to some $y \in W$. Therefore we have by assumption  on $(s_\alpha)$ that
\[
(x+s,y+t) +\cL = \lim_\alpha \, (x-s_\alpha, y)+\cL =\lim_\alpha \, (x_\alpha, y)+\cL = \lim_\beta \, (0, y-x_\beta^\star)+\cL = (0,0)+\cL \ ,
\]
where last claim follows from $\lim_\beta \,  (y-x_\beta^\star) =0$. This shows $(x+s,y+t) \in \cL $. Therefore $x+s \in \pi^G(\cL)$ and $(x+s)^\star = y+t \in W+t$, which proves the second inclusion.

\smallskip

\noindent For the first inclusion let $x \in \Lambda_{U+t}$. Then $x^\star -t \in U$. Fix a compact set $K\subseteq G$ such that $x+s$ is an interior point of $K$. Then $(s_\alpha+\Gamma)\cap K =\Gamma' \cap K$ for all $\alpha > \alpha_{K}$.
Now pick compact sets $K_1 \subseteq G$ and $K_2 \subseteq H$ such that $\cL +(K_1 \times K_2)=G \times H$. For any $\alpha>\alpha_K$ we can then write
\[
(x-s-s_\alpha, x^\star - t)=(k^1_\alpha , k^2_\alpha)+(l_\alpha, l_\alpha^\star) \ ,
\]
where $(k^1_\alpha, k^2_\alpha) \in K_1\times K_2$ and $(l_\alpha, l_\alpha^\star) \in \cL$. By compactness, there exists some subnet $((k^1_\beta , k^2_\beta))$ in $K_1 \times K_2$ that is convergent to some $(k_1,k_2) \in K_1 \times K_2$. Now $((k^1_\beta,k^2_\beta)+\cL) =  ((x-s-s_\beta, x^\star - t)+\cL)$ converges to $(k_1,k_2)+\cL$ and, by assumption, also to $(x,x^\star)+\cL=(0,0)+\cL$. Thus $k_1 \in \pi^G(\cL)$ and $k_2=k_1^\star$, and we can write
\[
(x-s-s_\beta+k_1-k^1_\beta, x^\star - t+k_2-k^2_\beta)=(l_\beta+k_1, l_\beta^\star+k_1^\star) \,.
\]
Recall that $(x^\star - t+k_2-k^2_\beta)$ converges to $x^\star -t \in U$ and that $U$ is open. Thus there exists some index $\beta_0$ such that $x^\star - t+k_2-k^2_\beta \in U$ for all $\beta > \beta_0$. Hence we have
\[
x-s+k_1-k^1_\beta \in s_\beta+ \Lambda_{U} \subseteq  s_\beta+\Gamma
\]
for all $\beta > \beta_0$. Moreover, since $x-s+k_1-k^1_\beta$ converges to $x-s$, and since $x-s$ is an interior point of $K$, there exists some $\beta_1 >\beta_0$ such that
\[
x-s+k_1-k^1_\beta \in (s_\beta+\Gamma) \cap K = \Gamma'\cap K  \subseteq \Gamma'
\]
for all $\beta >\beta_1$.
As $x-s+k_1-k^1_\beta \in \Gamma'$ converges to $x-s$ and $\Gamma'$ is closed, we get $x-s \in \Gamma'$, as claimed.
\end{proof}

The previous lemma has, in conjunction with Proposition~\ref{prop2}, the following consequence.

\begin{prop}\label{prop:hull-reg-is-reg} Let $\Lambda\subseteq G$ be a point set and let $\Gamma \in \XX(\Lambda)$.
\begin{itemize}
  \item[(a)] If $\Lambda$ is a (regular) model set, then $\Gamma$ is a (regular) model set.
  \item[(b)] If $\Lambda$ is an inter (regular) model set, then $\Gamma$ is an inter (regular) model set.
  \item[(c)] If $\Lambda$ is an almost (regular) model set, then $\Gamma$ is an almost (regular) model set.
\end{itemize}
\end{prop}

\begin{proof}
``(a)'': Let $(G,H,\cL)$ be a cut-and-project scheme and let $W\subseteq H$ be measure-theoretically regular such that $\Gamma=\Lambda_{W}$. Then there exists $(s,t)\in G\times H$ such that $\Lambda_{t+W^\circ} \subseteq s+\Gamma \subseteq \Lambda_{t+\overline{W}}$.
As  $s+\Gamma \subseteq \Lambda_{t+\overline{W}} \subseteq \pi^G(\cL)$, we can define
\begin{displaymath}
W'= (t+W^\circ) \cup (s+\Gamma)^\star \ .
\end{displaymath}
Clearly $W'$ has interior points by assumption on $W$. Note that $W'$ is measure-theoretically regular, as $(t+W^\circ) \subseteq W' \subseteq (t+\overline{W})$ implies $\partial W' \subseteq t+ \partial W$. We show $s+\Gamma = \Lambda_{W'}$, which yields the claim of the proposition by Theorem~\ref{thm:shift}.

\noindent
The inclusion $s+\Gamma \subseteq \Lambda_{W'}$ is obvious. In order to prove $\Lambda_{W'}\subseteq s+\Gamma$, consider $g \in \Lambda_{W'}$ such that $g^\star \in W'$. If $g^\star \in t+W^\circ$, then clearly $g \in \Lambda_{t+W^\circ}\subseteq s+\Gamma$. If $g^\star \in (s+\Gamma)^\star$, then there exists some $\gamma \in s+\Gamma$ such that $g^\star = \gamma^\star$ and hence $\gamma-g \in \ker(\star)$.
Since $s+\Gamma \in  \XX(\Lambda)$, there exists a net $(s_\alpha)$ in $G$ such that $(s_\alpha+\Lambda)$ converges to $s+\Gamma$ in the vague topology.
As $\Lambda = \Lambda_{W}$ is $ \ker(\star)$-periodic, so is $s_\alpha+\Lambda$ for all $\alpha$. Thus also the limit $s+\Gamma$ is $ \ker(\star)$-periodic. But this implies $g\in s+\Gamma$, which proves the claim.

\noindent ``(b)'', ``(c)'': By Lemma~\ref{lem:tp}, if $\Lambda$ is an inter/almost (regular) model set then $\Gamma$ is a shifted  inter/almost (regular) model set. The claim follows from Proposition~\ref{prop2}.
\end{proof}

Under slightly stricter assumptions, we can conclude even more. Let us first prove the following lemma.

\begin{lemma} Fix some cut-and-project scheme $(G, H, \cL)$, and take some non-empty open $U \subseteq H$ and compact $W\subseteq H$ such that $U\subseteq W$.  Consider $\Gamma\subseteq G$ such that $\Lambda_{U} \subseteq \Gamma \subseteq \Lambda_{W}$. If $W \backslash U$ is nowhere dense, then there exist some $\Gamma' \in \XX(\Lambda)$ and some $t \in H$ such that $\Gamma'=\Lambda_{t+U} = \Lambda_{t+W}$. In particular, $\Gamma'$ is a model set.
\end{lemma}

\begin{proof}
As usual, by eventually replacing $H$ by the subgroup generated by an open pre-compact set $K$ which contains $W$, we can assume without loss of generality that $H$ is compactly generated. Thus $H$ can be assumed $\sigma$-compact without loss of generality.
By \cite[Prop.~8.3]{KR19} there exists some $t \in H$ such that $(t+ (W \backslash U)) \cap \pi^{H}(\cL)= \varnothing$.
Now choose a net $((s_\alpha, s_\alpha^\star))$ in $\cL$ such that  $s_\alpha^\star \to t$, which is possible by denseness. By compactness, $(s_\alpha+\Gamma)$ has a subnet $(s_\beta+\Gamma)$ convergent to some $\Gamma' \in \XX(\Lambda)$.
Thus $(s_\beta+\Gamma)$ converges to $\Gamma'$, and $((-s_\beta, 0)+\cL)=((0, s_\beta^\star)+\cL)$ converges to $(0,t)+\cL$. Therefore Lemma~\ref{lem:tp} yields $\Lambda_{U+t} \subseteq \Gamma' \subseteq \Lambda_{W+t} $.
As $(t+ (W \backslash U)) \cap \pi^{H}(\cL)= \varnothing$, the claim follows.
\end{proof}

Since any set of zero Haar measure is nowhere dense, we have, in conjunction with Theorem~\ref{thm:shift}, the following result.

\begin{cor}\label{cor:hull-rms}
Let $\Gamma$ be an almost regular model set. Then there exists some $\Gamma' \in \XX(\Gamma)$ which is a regular model set. As a consequence, if $\Gamma$ is an inter regular model set, then there exists some $\Gamma' \in \XX(\Gamma)$ which is a regular model set.
\end{cor}

\medskip

We finish the section with a nice consequence for repetitive point sets. Recall that a point set $\Lambda$ of finite local complexity (such as a subset of a model set) is called repetitive if the set
\[
\{ t \in G : (-t+\Lambda) \cap K = \Lambda \cap K \}
\]
is relatively dense in $G$, for any compact set $K$. For point sets of finite local complexity, repetitivity of $\Lambda$ is equivalent to $\XX(\Lambda)$ being minimal, when equipped with its natural translation action \cite[Fact~3]{BLM}. By combining the results of this section, we get:

\begin{prop}\label{prop:flcrep} Let $\Gamma\subseteq G$ be a point set of finite local complexity that is repetitive. Then, the following are equivalent:
\begin{itemize}
  \item[(i)] $\Gamma$ is a regular model set.
  \item[(ii)] $\Gamma$ is an inter regular model set.
  \item[(iii)] $\Gamma$ is an almost regular model set.
\end{itemize}
\end{prop}
\begin{proof}
(i)$\Longrightarrow$(ii) $\Longrightarrow$ (iii) is obvious.

\noindent (iii)$\Longrightarrow$ (i): By Corollary~\ref{cor:hull-rms}, there exists some $\Gamma' \in \XX(\Gamma)$ which is a regular model set. By minimality and Proposition~\ref{prop:hull-reg-is-reg}, we get that $\Gamma \in \XX(\Gamma')$ is a regular model set.
\end{proof}

\section{Second countable internal spaces}

Consider a cut-and-project scheme for some almost (regular) model set. Here we show that the internal space may be assumed second countable. In fact we will prove a slightly stronger statement, which is summarised as follows.

\begin{prop}\label{mainsecond} Let $(G,H,\cL)$ be a cut-and-project scheme. Let $K\subseteq H$ be compact and let $U \subseteq K$ be measurable. Then there exist a cut-and-project scheme $(G, H', \cL')$ where $H'$ is second countable, a compact set $K'\subseteq H'$, an open set $U' \subseteq  K'$ and some positive constant $c$ such that
\begin{align*}
c \cdot m_H(U) \le m_{H'}(U')&\le m_{H'}(K')\leq c \cdot m_H(K)
\end{align*}
and $\Lambda'_{U'} \subseteq \Lambda^{}_U \subseteq \Lambda^{}_K \subseteq \Lambda'_{K'} $.
In particular, any almost (regular) model set is an almost (regular) model set in some cut-and-project scheme with second countable internal space.
\end{prop}

The proof of Proposition \ref{mainsecond} will be given below. It is based on the following lemma.

\begin{lemma}\label{lem:metcomp} Let $H$ be a locally compact abelian group with topology $\tau$. Let  $U\subseteq H$ be any zero neighborhood, and let $\mathcal{D}= \{f_i: i \in \NN\}\subseteq \Cc(H)$ be any countable uniformly bounded family of  functions.
Then there exist a metrisable group $(\widetilde{H}, \widetilde d)$ and a continuous group homomorphism $\phi : (H,\tau) \to (\widetilde{H}, \widetilde d)$ with dense range such that:
\begin{itemize}
  \item[(a)] For every $f \in \mathcal{D}$ there exists $\widetilde f \in \Cc(\widetilde{H})$ such that $ f= \widetilde f \circ \phi$.
  \item[(b)] There exists some zero neighborhood $\widetilde V\subseteq \widetilde H$ such that $\phi^{-1}(\widetilde V) \subseteq U$.
\end{itemize}
\end{lemma}
\begin{proof}

Pick some $f_0 \in  \Cc(H)$ such that $\supp(f_0) \subseteq U$ and $f_0(0)=1$. We then have
$\{ s \in H : \|T_sf_0-f_0 \|_\infty < 1 \} \subseteq U$, where  $(T_sf)(h)=f(h+s)$. For $n\in\NN_0$ define translation invariant pseudometrics $d_n$ on $H$ by $d_n(s,t)= \| T_sf_n -T_tf_n \|_\infty$.
Then the family $(d_n)_{n\in\NN_0}$ is uniformly bounded, since $\mathcal D$ is uniformly bounded. Therefore, we can define a translation invariant pseudometric $d$ on $H$ via $d= \sum_{n\in\NN_0} 2^{-n} d_n$.  Then $Id: (H,\tau) \to (H,d)$ is continuous.  Indeed, since any finite set of functions in $\Cc(H)$ is equi-uniformly continuous, for each $\varepsilon >0$ there exists some zero neighborhood $V\subseteq H$ such that $V \subseteq \{ t \in H : d(t,0) < \varepsilon \}$. Let $\phi': (H,d) \to (\widetilde H, \widetilde d)$ be the completion map and define $\phi: (H,\tau)\to (\widetilde H, \widetilde d)$ by $\phi=\phi'\circ Id$. By construction $\widetilde H$ is a locally compact abelian group, and $\phi$ is a continuous group homomorphism with dense range, see e.g.~\cite[Ch. III, Thm. I, Thm. II]{BOU}. We verify properties (a) and (b).

\smallskip
\noindent ``(a)'' Consider arbitrary $n \in \NN_0$. We show that $f_n$ is uniformly continuous on $(H,d)$. Let $\varepsilon >0$ and define $\delta=2^{-n}\varepsilon$. If $d(s,t) \le\delta$ for some $s,t\in H$, we then have $2^{-n}d_n(s,t)\le \delta$ and hence
\[
|f_n(s)-f_n(t)| \le \|T_sf_n-T_tf_n\|_\infty = d_n(s,t)\le  2^n \delta = \varepsilon \,.
\]
As $f_n$ is uniformly continuous on $(H,d)$, it can be extended uniquely to the completion, i.e., there exists some $\widetilde f_n \in \Cu(\widetilde H)$ such that $\widetilde f_n \circ \phi = f_n$. We claim that $\widetilde f_n$ has compact support. Define $K_n=\supp(f_n)$ and consider the set  $\phi(K_n)\subseteq \widetilde H$, which is compact since $\phi$ is continuous. It suffices to show that $\widetilde f_n(\widetilde x)\ne0$ implies $\widetilde x \in \phi(K_n)$. Thus consider $\widetilde x \in \widetilde H$ such that $\widetilde f_n(\widetilde x) \neq 0$. By denseness of $\phi(H)$ in $\widetilde H$ and by metrisability of $\widetilde H$, we can find a sequence $(x_k)_{k \in \NN}$ in $H$ such that $\phi(x_k)\to \widetilde x$ as $k\to\infty$. We then have $\widetilde f_n(\phi(x_k)) \to \widetilde f_n(\widetilde x)$ as $k\to\infty$.  Since $\widetilde f_n(\widetilde x) \neq 0$, there exists some $k_0\in \NN$ such that for all $k>k_0$ we have $f_n(x_k)=\widetilde f_n(\phi(\widetilde x_k)) \neq 0$. This shows that $x_k \in K_n$ for all $k >k_0$ and hence $\phi(x_k) \in \phi(K_n)$ for all $k >k_0$. Since $\phi(K_n)$ is compact and $\phi(x_k) \to \widetilde x$ as $k\to\infty$, we get $\widetilde x \in \phi(K_n)$.

\smallskip

\noindent ``(b)'' We show that $\widetilde V= \{\widetilde s \in \widetilde H: \widetilde d(\widetilde s, \widetilde 0)<1\}$ has the desired property. Indeed, let $s \in \phi^{-1}(\widetilde V)\subseteq H$. Then $\phi(s) \in \widetilde V$ gives $ \|T_sf_0-f_0 \|_\infty =d_0(s,0) \leq d(s,0) = \widetilde d(\phi(s),\widetilde 0)< 1$. Thus $s\in U$, which proves the claim.
\end{proof}

\begin{remark}[On injectivity of $\phi$]
From the proof of Lemma~\ref{lem:metcomp} we infer
\begin{displaymath}
\begin{split}
\ker(\phi) &= \{ s \in H : d(s,0)=0 \} = \bigcap_{n\in\NN_0} \left\{s \in H : d_n(s,0) =0 \right\} = \bigcap_{n\in\NN_0} \text{Per}(f_n) \ ,
\end{split}
\end{displaymath}
where $\text{Per}(f_n)$ denotes the set of periods of $f_n$.
%
%
In particular, $\phi$ can be chosen injective if $H$ is metrisable. Indeed take $\psi \in \Cc(H)$ such that $0 \leq \psi \leq 1$ and $\psi(0)=1$, which is possible by Urysohn's lemma. Take any metric $d$ compatible with the topology of $H$ and define $\varphi\in \Cc(H)$ by  $\varphi(h)=\psi(h)\cdot  \left(1- d(h,0) \right)$. Then  $\varphi$ is aperiodic since $\varphi(h)=1$ if and only if $h=0$.
\end{remark}

\begin{proof}[Proof of Proposition~\ref{mainsecond}]

By regularity of the Haar measure $m_H$ on $H$ and by Urysohn's lemma, we can find an increasing sequence $(f_n)_{n\in\NN}$ in $\Cc(H)$ and a decreasing sequence $(g_n)_{n\in \NN}$ in $\Cc(H)$ such that $\|g_n \|_\infty =1$, $m_H(g_n-1_K)\le 1/n$, $m_H(1_U-f_n) \leq \frac{1}{n}$ and $0 \leq f_n \leq 1_{U} \leq 1_{K} \leq g_n\le 1$ for all $n\in \NN$. Since $\cL$ is lattice, there exists an open zero neighborhood $O_G\times O_H \subseteq G\times H$ such that $\cL \cap (O_G \times O_H) = \{ (0,0) \}$. We apply Lemma~\ref{lem:metcomp} with $\mathcal{D}:= \{ f_n,g_n: n\in \NN \}$ and $U=O_H$. Therefore, there exist a group $\widetilde H$ and some group homomorphism $\phi: H \to \widetilde H$ satisfying the conclusion of Lemma~\ref{lem:metcomp}. Define  $\cL' := \{ (x, \phi(x^\star)) : x \in L \}$ and $H'=\widetilde H$. We split the rest of the proof into steps.

\smallskip

\noindent {\it Step 1:  We show that $(G, H', \cL')$ is a cut-and-project scheme.} First we show that the group $\cL'$ is a lattice. By Lemma~\ref{lem:metcomp}, there exists an open zero neighborhood $O_{H'} \subseteq H'$ such that $\phi^{-1}(O_{H'}) \subseteq O_{H}$. It follows immediately that $\cL' \cap (O_G \times O_{H'}) = \{ (0,0) \}$. Hence $\cL'$ is uniformly discrete. In order to show relative denseness of $\cL'$, note that there exist compact sets $K_G\subseteq G$ and $K_H\subseteq H$ such that $\cL+(K_G \times K_H)= G \times H$. Then, by continuity of the completion map we have that $\phi (K_H)$ is compact and $\cL'+(K_G \times \phi(K_H))= G \times \phi(H)$. Since $\phi(H)$ is dense in $H'$, there exists a compact sets $K_{H'}$ with non-empty interior such that $\phi(H)+K_{H'}=H'$ and hence $\cL'+(K_G \times (\phi(K_H) +K_{H'}))= G \times H'$. Hence the group $\cL'$ is a Delone set, i.e., $\cL'$ is a lattice.
The lattice $\cL'$ projects injectively to $G$ as $\cL$ does and as $\phi(0^\star)=0\in H'$. In order to show that $\cL'$ projects densely to $H'$, let $O$ be any non-empty open set in $H'$. Then $\phi^{-1}(O)\subseteq H$ is open by continuity of the completion map, and $\phi^{-1}(O)$ is non-empty by denseness of $H$ in $H'$. Therefore there exists some $x\in L$ such that $x^\star \in \phi^{-1}(O)$. But this means  $\phi(x^\star) \in O$.

\smallskip

\noindent {\it Step 2: We construct the sets $U'$ and $K'$.} By Lemma~\ref{lem:metcomp} there exist $f'_n, g'_n \in \Cc(H')$ such that $f_n = f'_n \circ \phi$ and $g_n = g'_n \circ \phi$. Denseness of $\phi$ combined with $0\leq f_1 \leq \ldots \leq f_n \leq \ldots \leq g_n \leq \ldots \leq g_1\le 1$ gives $0 \leq f'_1 \leq \ldots \leq f'_n \leq \ldots \leq g'_n \leq \ldots \leq g'_1\leq 1$. Now define
\begin{displaymath}
  U'= \{ h \in H' :  f'_n(h) > 0 \text{ for some } n\in\NN \} \ , \qquad K'= \{ h \in H' : g'_n(h)  = 1 \text{ for all } n\in\NN \} \ .
\end{displaymath}
By definition $U'$ is open in $H'$, and $K'$ is a closed subset of $\supp(g'_1)$ and hence compact in $H'$. We show that $U' \subseteq K'$. Let $h \in U'$. Then, since $\phi(H)$ is dense in $H'$ and $U'$ is open, there exists a sequence $(x_k)_{k\in \NN}$ in $H$ such that $\phi(x_k) \in U'$ and $\phi(x_k) \to h$ as $k\to\infty$. Note that, by construction, for every $k$ there exists some $n=n(k)$ such that   $0 < f'_n(\phi(x_k))= f_n(x_k) \leq 1_{U}(x_k)$. In particular this yields $x_k \in U \subseteq K$ and hence $g'_m(\phi(x_k))=g_m(x_k) = 1$ for all  $m\in \NN$. This gives $\phi(x_k) \in K'$ for all $k\in \NN$. As $\phi(x_k) \to h$ and $K'$ is compact, we get  $x \in K'$.
This proves $U' \subseteq K'$.

\smallskip

\noindent {\it Step 3: Fix a Haar measure $m_{H'}$ on $H'$. We prove that there exists some constant $c \in (0, \infty)$ such that $m_{H'}(K')\leq c \cdot m_H(K)$ and $m_{H'}(U')\geq c \cdot m_H(U)$.}  Consider the pushforward of the Haar measure $m_H$ on $H$ through $\phi$. This is an $H$-invariant measure on $H'$ and hence, by denseness, an $H'$-invariant measure on $H'$.
Therefore there exists a constant $d \in [0,\infty)$ such that for all $f'\in \Cc(H')$ we have
\[
\int_H f' \circ \phi \,  \dd m_H = d \cdot  \int_{H'} f' \dd m_{H'} \ .
\]
Applying this formula to $f'=g'_1$, we get $d\cdot m_{H'}(g'_1)= m_H(g'_1\circ \phi)=m_H(g_1)>0$, which gives $d \neq 0$. Therefore we have for all $n\in \mathbb N$ that
\begin{displaymath}
d\cdot m_{H'}(K') \leq d \cdot \int_{H} g_n'(t) \dd m_H =\int_H g_n' \circ \phi \,  \dd m_H  =\int_H g_n \,  \dd m_H   \,.
\end{displaymath}
Letting $n \to \infty$ we get $m_{H'}(K')\leq  m_H(K)/d$.

Next, $0 \leq f_n \leq 1_{U}$ gives that $0 \leq f_n'(x') \leq 1$ on the dense subset $\phi(H)\subseteq H'$ and hence, via continuity, $0 \leq f'_n(x') \leq 1$ for all $n\in \mathbb N$ and $x'\in H'$. The definition of $U'$ then gives $0 \leq f_n' \leq 1_{U'}$. Therefore
\begin{displaymath}
d\cdot m_{H'}(U') \geq d \cdot \int_{H} f_n'(t) \dd m_H =\int_H f_n' \circ \phi \,  \dd m_H  =\int_H f_n \,  \dd m_H   \,.
\end{displaymath}
Letting $n \to \infty$ we get $m_{H'}(U')\geq m_H(U)/d$. The claim now follows with $c=1/d$.

\smallskip

\noindent {\it Step 4: We prove $\Lambda'_{U'} \subseteq \Lambda^{}_U$ and  $\Lambda^{}_K \subseteq \Lambda'_{K'}$.}
Let $x \in \Lambda'_{U'}$. Then $\phi(x^\star) \in U'$, and hence there exists some $n\in\NN$ such that
$f_n(x^\star)= f'_n(\phi(x^\star)) > 0$.
Since $f_n \leq 1_{U}$ we get that $x^\star \in U$ and hence $x \in \Lambda_U$. This shows
$\Lambda'_{U'} \subseteq \Lambda_U$. If $x \in \Lambda_K$ we have $x^\star \in K$ and hence for all $n\in\NN$ we have $g'_n(\phi(x^\star))=g_n(x^\star) = 1$. This means  $\phi(x^\star) \in K'$ and $x^\star \in \Lambda_{K'}$.

\smallskip

It follows immediately that any almost model set $\Lambda_U \subseteq \Gamma \subseteq \Lambda_K$ in $(G, H, \cL)$ satisfies $\Lambda_{U'} \subseteq \Gamma \subseteq \Lambda_{K'}$ and hence is an almost model set in $(G, H', \cL')$. If furthermore $m_H(K \backslash U)=0$, then step 3 gives $m_{H'}(K' \backslash U')=0$, and hence measure-theoretic regularity is preserved.

\noindent At this point we have proved the proposition -- apart from $\sigma$-compactness of $H'$. But $\sigma$-compactness may be assumed without loss of generality, as $K'$ lies inside a compactly generated open and closed subgroup of $H'$.
\end{proof}

\section{Extended cut-and-project schemes}\label{sec:ecps}

In this section we restrict to cut-and-project schemes $(\RR^d, H, \cL)$ with Euclidean direct space and second countable internal space $H$. Thus $\cL$ is countable, and also the annihilator $\cL^\circ\subseteq \RR^d \times \widehat H$ of $\cL$ is countable.

\begin{lemma} \label{lem2}
Let $A$ be any countable subset of $\RR^d$. Then there exists a lattice $D$ in $\RR^d$ such that $D \cap \Qspan(A)=\{0\}$ and $D^\circ \cap \Qspan(A)=\{0\}$, where $D^\circ=\{x\in \RR^n: \langle d,x\rangle \in \ZZ \text{ for all } d\in D\}$ is the lattice dual to $D$.
\end{lemma}

\begin{proof}
We will use the coordinate projections $\pi_i:\RR^d\to \RR$, given by $(x_1,\ldots, x_d)^T\mapsto x_i$.  Define the set $S=\bigcup_{i=1}^d \pi_i(A)\subseteq \RR$ of all entries of vectors from $A$, and consider the
$\QQ$-vector space $V=\Qspan(S)\subseteq \RR$. Let $T$ be a countable $\QQ$-basis for $V$. We claim that there exist nonzero real numbers $c_1, \ldots, c_d$ such that the set
$$
T'=T \cup \{ c_1, \ldots, c_d \} \cup \{c_1^{-1},\ldots, c_d^{-1} \}
$$
is $\QQ$-linearly independent. This is seen inductively. Indeed, if $ T \cup \{ c, \frac{1}{c} \}$ is $\QQ$-linearly dependent, then $c$ satisfies a quadratic equation over the countable multivariate polynomial ring $\QQ[T]$. As the number of quadratic equations with coefficients in $\QQ[T]$ is countable,  hence so is the number of roots of such equations. Since $\RR$ is uncountable, we can pick a $c \in \RR$ which is not the root of such an equation.

Consider now the lattice $D=\Zspan(\{c_1e_1, \ldots, c_de_d\})\subseteq \RR^d$, where $e_i\in \RR^d$ is the $i^{th}$ coordinate vector, i.e., we have $\pi_j(e_i)=\delta_{ij}$ for $i,j\in \{1,\ldots, d\}$. In order to show $D \cap \Qspan(A)=\{0\}$, consider any $u\in D \cap \Qspan(A)$.
Then there exist vectors  $a_1, \ldots, a_k \in A$, rationals $q_1, \ldots, q_k \in \QQ$ and integers $n_1,\ldots , n_d \in \ZZ$ such that
\begin{displaymath}
  u=  \sum_{i=1}^d n_i c_i e_i = \sum_{j=1}^k q_j a_j \ .
\end{displaymath}
For the vector components we have $n_ic_i= \sum_{j=1}^k q_j \pi_i (a_j)$ for $i\in\{1,\ldots, d\}$. Using $\QQ$-linear independence of $T'$ we conclude $n_1=\ldots = n_d=0$ and $q_1=\ldots=q_k=0$, which proves $D \cap \Qspan(A)=\{0\}$. The claim $D^\circ \cap \Qspan(A)=\{0\}$ is shown analogously, noting $D^\circ=\Zspan(\{c_1^{-1}e_1, \ldots, c_d^{-1}e_d\})\subseteq \RR^d$.
\end{proof}

\medskip

\begin{lemma}\label{lem3} Let $(\RR^d, H, \cL)$ be a cut-and-project scheme such that $\cL$ and $\cL^\circ$ are both countable. Let $D$ be a lattice in $\RR^d$ such that $D \cap \pi^{G}(\cL)=\{0\}$ and $D^\circ \cap \pi^{\hat{G}}(\cL^\circ)=\{0\}$. Consider the compact abelian group $\KK= \RR^d / D$ with quotient map $\psi: G \to \KK$. Then the following hold.
\begin{itemize}
  \item[(a)] $\psi$ is one-to-one on $\pi^{G}(\cL)$, and $\psi(\pi^{G}(\cL))$ is dense in $\KK$.
  \item[(b)] For each non-empty open set $U \subseteq H$ the set $\psi(\Lambda_U)$ is dense in $\KK$.
\end{itemize}
\end{lemma}

\begin{remark}
As the following proof shows, the construction in Lemma~\ref{lem3} works more generally for cut-and-project schemes with second countable $G$ and $H$, if  there is a continuous group homomorphism to a compact metrisable group $\KK$, such that $\ker(\psi)\cap \pi^G(\cL)=\{0\}$ and $\ker(\widehat \psi)\cap \pi^{\widehat G}(\cL^0)=\{0\}$, where $\widehat \psi$ denotes the group homomorphism dual to $\psi$. In the setting of the lemma, the latter condition may be satisfied due to Lemma~\ref{lem2}.
\end{remark}

\begin{proof}
Existence of a lattice $D\subset \RR^d$ such that $D \cap \pi^{G}(\cL)=\{0\}$ and $D^\circ \cap \pi^{\hat{G}}(\cL^\circ)=\{0\}$ follows from Lemma~\ref{lem2}.

\noindent (a) Injectivity of $\psi|_{\pi^{G}(\cL)}$ is immediate from $\pi^{G}(\cL) \cap D =\{ 0\}$. Denseness is a consequence of Pontryagin duality, compare \cite[Prop.~4.1]{RVM3}: The map $\alpha: \cL \to \KK$, given by $\ell\mapsto \psi(\pi^G(\ell))$, is a continuous group homomorphism. Thus $\psi(\pi^{G}(\cL))$ is dense in $\KK$ if the dual group homomorphism $\widehat\alpha: D^\circ\to (\widehat G\times \widehat H)/\cL^\circ$, given by $\chi \mapsto (\chi,1)+\cL^\circ$, is one-to-one. Here $1\in \widehat H$ denotes the trivial character. But the latter holds as $D^\circ \cap \pi^{\hat{G}}(\cL^\circ)=\{0\}$.

\noindent (b)
Here we follow ideas behind Weyl's uniform distribution theorem. Consider the space $\mathcal M_f(\KK)$ of finite Borel measures on $\KK$, equipped with the topology of weak convergence. Fix the van Hove sequence $(A_n)_{n\in \mathbb N}$ in $\RR^d$ where $A_n=[-n,n]^d$, and let $V \subseteq U\subseteq H$ be a compact measure-theoretically regular set having non-empty interior. Consider measures $\mu_n\in \mathcal M_f(\KK)$ given by
\begin{displaymath}
\mu_n= \frac{1}{m_G(A_n)} \sum_{x \in A_n \cap \Lambda_V} \delta_{\psi(x)} \ ,
\end{displaymath}
where $m_G$ is the Lebesgue measure.
Then $(\mu_n)_{n\in\mathbb N}$ is a sequence of equi-bounded measures in $\mathcal M_f(\KK)$. We claim that $(\mu_n)_{n\in\mathbb N}$ converges weakly to $\dens(\Lambda_V) \cdot m_{\KK}$, where $m_{\KK}$ is the Haar probability measure on $\KK$.

Since $\KK$ is metrisable, any set of equi-bounded measures in $\mathcal M_f(\KK)$ is compact and metrisable. Therefore it suffices to show that any convergent subsequence converges vaguely to $\dens(\Lambda_V) \cdot m_{\KK}$.
Let $(\mu_{k_n})_{n\in\mathbb N}$ be a subsequence of $(\mu_n)_{n\in\mathbb N}$ that converges vaguely to $\mu$. Consider the Fourier-Stieltjes transform $\widehat \mu$ of $\mu$, which is a uniformly continuous and bounded function on $\widehat{\KK}\cong D^\circ$. Then, for all $\chi \in  D^0$ we have by weak convergence and $D$-invariance that
\[
\widehat\mu(\chi)=\int_{\KK} \chi(t) \dd \mu(t) = \lim_n \int_{\KK} \chi(t) \dd \mu_{k_n}(t) = a_{\chi}( \delta_{\Lambda_V}) \ ,
\]
where $a_\chi(\nu)=\lim_{n\to\infty} \frac{1}{|A_{k_n}|} \int \overline{\chi}(t) {\rm d}\nu(t)$ denotes the Fourier-Bohr coefficient of $\nu$ at $\chi$ along $(A_{k_n})_n$. Existence and properties of Fourier-Bohr coefficients for regular model sets have been discussed e.g.~in \cite[Sec.~7]{L09} and in \cite[Cor.~3.40]{LSS20}. For $\chi=1$ we have $a_{1}(\delta_{\Lambda_V})= \dens(\Lambda_V)$.
By construction, for $\chi\ne 1$ we have $\chi \notin  \pi^{\hat{G}}(\cL^\circ)$ and hence
$a_{\chi}(\delta_{\Lambda_V})= 0$.
Obviously, the measure $\mu'=\dens(\Lambda_V)\cdot  m_{\KK}$ has the same Fourier-Stieltjes transform as $\mu$.
By injectivity of the Fourier transform \cite[1.7.3 (b)]{R62} we conclude $\mu=\mu'$.

We finally show that $\psi(\Lambda_V)$ is dense in $\KK$. Consider an arbitrary non-empty open set $O\subseteq \KK$. Take $\phi \in \Cc(\KK)$ such that $\supp(\phi) \subseteq O$ and $m_\KK(\phi) \neq 0$. We then have $0\neq \mu(\phi)=\lim_{n\to\infty} \mu_n(\phi)$. Take some $n$ such that $0\neq \mu_n(\phi)$.
Then there exists some $t \in \supp(\phi) \cap \psi (\Lambda_V \cap A_n) \subseteq O \cap \psi(\Lambda_V)$.
\end{proof}

\begin{lemma}\label{lem4} Under the conditions of Lemma~\ref{lem3}, the map
$\phi: \pi^G(\cL) \to H \times \KK$, given by $\phi(x)=(x^\star, \psi(x))$, has dense range.
\end{lemma}

\begin{proof}
Let $ \varnothing \neq U \subseteq H \times \KK$ be any open set. Then, there exist open sets $V \subseteq H$ and $O \subseteq \KK$ such that $ \varnothing \neq V \times O \subseteq U$. By Lemma~\ref{lem3}, the set $\psi(\Lambda_V)$ is dense in $\KK$. Therefore, there exists some $x \in \Lambda_V$ such that $\psi(x) \in O$. Thus $\phi(x) =(x^\star, \psi(x)) \in V \times O \subseteq U$.
\end{proof}

We can now define an extended cut-and-project scheme with injective $\star$-map.

\begin{theorem}[Extended cut-and-project scheme]\label{thm2} Let $(G, H, \cL)$ be a cut-and-project scheme with $G=\RR^d$ such that $\cL$ and $\cL^\circ$ are both countable.
Let $\KK$ be as in Lemma~\ref{lem3} and define $H'=H \times \KK$. Let $\phi:\pi^G(\cL)\to H'$ be as in Lemma~\ref{lem4} and define $\cL'= \{ (x, \phi(x)) : x \in \pi^G(\cL) \}$. Then the following hold.
\begin{itemize}
  \item[(a)] $(G, H', \cL')$ is a cut-and-project scheme with injective $\star$-map $\phi$.
  \item[(b)] For any $W \subseteq H$ we have $\Lambda_W=\Lambda'_{W'}$ where $W'=W \times  \KK$.
  \item[(c)] If $H$ is $\sigma$-compact or metrisable, respectively, then so is $H'$.
\end{itemize}
\end{theorem}

\begin{remark}\label{rem:preserved}
Note that, under the above assumptions, any model set in $(G, H, \cL)$ is a model set in $(G, H', \cL')$, and the window associated to the extended cut-and-project scheme inherits measurability and any property in $(T)$ or $(M)$.
\end{remark}

\begin{proof}

\noindent ``(a)" We show that the subgroup $\cL'$ of $G \times H'$ is a lattice. Since $\cL$ is a lattice in $G \times \KK$, there exists open sets $0 \in U \subseteq G$ and $0 \in V \subseteq H$ such that
$\cL \cap (U \times V)= \{ (0,0) \}$. Then $\cL' \cap (U \times V \times \KK)= \{ (0,0,0) \}$. Hence $\cL'$ is discrete. As to relative denseness, let $K \subseteq G$ and $W \subseteq H$ be compact sets so that $\cL+(K \times W) = G \times H$.  Then $\cL'+(K \times W \times \KK)= G\times H \times \KK$. Turning to the projection assumptions, injectivity of $\pi^G|_{\cL'}$ readily follows from injectivity of $\pi^G|_\cL$. Finally, $\pi^{H'}(\cL')$ is dense in $H'$ by Lemma~\ref{lem4}. This proves (a).

\noindent ``(b)'' is obvious as $\Lambda'_{W'}=\pi^G(\cL' \cap(G\times W\times \KK))=\pi^G(\cL \cap(G\times W))=\Lambda_W$.

\noindent ``(c)'' follows immediately from the fact that the construction in Lemma~\ref{lem2} produces a compact and metrisable $\KK$ .
\end{proof}

\medskip

\begin{proof}[Proof of Theorem~\ref{thm:mainintro}]
Let $\Lambda \subseteq \RR^d$ be any shifted almost model set. Then $\Lambda$ is an almost model set by Theorem~\ref{thm:shift}, and hence it is an almost model set in a cut-and-project scheme with second countable internal space by Proposition~\ref{mainsecond}. Applying Theorem~\ref{thm2}, we get that there exist a cut-and-project scheme $(\RR^d, H' , \cL')$ with second countable $H'$ and injective $\star$-map, and $U' \subseteq K'$ such that $U'$ is non-empty open, $K'$ is compact and $\Lambda_{U'} \subseteq \Lambda \subseteq \Lambda_{K'}$. Measurability and properties $(T)$ and $(M)$ are inherited due to Remark~\ref{rem:preserved}.
Furthermore, if $\Lambda$ is an almost regular model set, we also have $m_{H'}(K' \backslash U')=0$.
Defining $W= U' \cup (\Lambda^\star)$ exactly as in Proposition~\ref{lem:imschar}, we get $\Lambda_W=\Lambda$, and the rest of the claim follows. The last statement follows from the observation that any inter (regular) model set is an almost (regular) model set by  Proposition~\ref{prop:interprojection}.
\end{proof}

\section*{Acknowledgements}
NS was supported by the Natural Sciences and Engineering Council of Canada (NSERC) via grant 2020-00038, and he is grateful for the support.
Part of this work was done during the "Summer School on Aperiodic Order" at MacEwan University which was sponsored by the Research Office, the Dean Office and the Mathematics and Statistics Department at MacEwan University, as well as Collaborative Research Centre 1283 of the DFG at the Bielefeld University, and the authors are grateful for the support.

\end{document}